\documentclass{eptcs}
\usepackage{listings}
\usepackage{algorithm}
\usepackage[noend]{algpseudocode}
\usepackage{amsmath}
\usepackage{hyperref}
\usepackage{cite}
\usepackage{boogie}
\usepackage{proof}
\usepackage{breakurl}
\usepackage{underscore}
\usepackage{amsthm}
\usepackage{xspace}

\newcommand{\powerfx}{\textit{Mod 1}}
\newcommand{\diskdump}{\textit{Mod 2}}
\newcommand{\hid}{\textit{Mod 3}}
\newcommand{\usbstor}{\textit{Mod 4}}
\newcommand{\mpiomsdm}{\textit{Mod 5}}
\newcommand{\mpiocontrol}{\textit{Mod 6}}
\newcommand{\gpio}{\textit{Mod 7}}
\newcommand{\scsiport}{\textit{Mod 8}}
\newcommand{\iscsi}{\textit{Mod 9}}
\newcommand{\usbhub}{\textit{Mod 10}}
\newcommand{\xusb}{\textit{Mod 11}}
\newcommand{\db}{\textit{Mod 12}}
\newcommand{\io}{\textit{Mod 13}}
\newcommand{\port}{\textit{Mod 14}}
\newcommand{\ksfilter}{\textit{Mod 15}}
\newcommand{\audio}{\textit{Mod 16}}

\newcommand{\varsp}{\#}
\newcommand{\asN}{aS_0}
\newcommand{\tagged}{\textit{tagged}}
\newcommand{\pt}[1]{pt(\m{#1})}
\newcommand{\m}[1]{\mathsf{#1}}
\newcommand{\ptf}[2]{pt(#1\m{.#2})}
\newcommand{\defaultVar}{\textit{defaultVar}}
\newcommand{\nonNullExprs}{\textit{nonNullExprs}}
\newcommand{\hashFunction}{\textit{hashFunction}}

\newcommand{\emp}{\{\}}

\newcommand{\Omit}[1]{}
\newcommand{\Null}{\textsc{Null}\xspace}

\newtheorem{theorem}{Theorem}
\newtheorem{lemma}[theorem]{Lemma}

\title{Precise Null Pointer Analysis Through\\Global Value Numbering}
\author{Ankush Das
\institute{Carnegie Mellon University\\Pittsburgh, PA, USA}
\email{ankushd@cs.cmu.edu}
\and
Akash Lal
\institute{Microsoft Research\\Bangalore, India}
\email{akashl@microsoft.com}
}

\def\titlerunning{Precise Null Pointer Analysis through Global Value Numbering}
\def\authorrunning{Ankush Das \& Akash Lal}

\begin{document}
\maketitle

\begin{abstract}
Precise analysis of pointer information plays an important role in many static
analysis techniques and tools today. The precision, however, must be balanced
against the scalability of the analysis. This paper focusses on improving the
precision of standard context and flow insensitive alias analysis algorithms
at a low scalability cost. In particular, we present a semantics-preserving
program transformation that drastically improves the precision of existing
analyses when deciding if a pointer can alias \Null. 
Our program transformation is based on Global Value Numbering, a scheme 
inspired from compiler optimizations literature. It allows even a flow-insensitive
analysis to make use of branch conditions such as checking if a pointer is \Null
and gain precision.
We perform experiments on real-world code to measure the overhead
in performing the transformation and the improvement in the precision of the
analysis. We show that the precision improves from 86.56\% to 98.05\%, while
the overhead is insignificant. 


{\bf Keywords:} Alias Analysis, Global Value Numbering,
Static Single Assignment
\end{abstract}

\section{Introduction}
Null-pointer exceptions directly affect software reliability because such
exceptions can bring down the application. Detecting and eliminating these bugs 
is an important step towards developing reliable systems. Static analysis tools that
look for null-pointer exceptions typically employ techniques based on \textit{alias
analysis} to detect possible aliasing between pointers. 
Two pointer-valued variables are said to \textit{alias} if they hold the same
memory location during runtime. Aliasing can be decided in two ways:
(a) \textit{may-alias}~\cite{andersen-thesis},
where two pointers are said to may-alias if they 
can point to the same memory location under some possible execution, 
and (b) \textit{must-alias}~\cite{steensgaard-popl96}, where 
two pointers must-alias if they always point to the same memory location
under all possible executions.
Because a precise alias analysis is undecidable~\cite{RamaUndec}
and even a flow-insensitive pointer analysis is NP-hard~\cite{HorwitzNP},
much of the research in the area plays on the precision-efficiency trade-off
of alias analysis. For example, practical algorithms for may-alias analysis lose precision
(but retain soundness) by over-approximating: a verdict that two pointer may-alias  does not imply
that there is some execution in which they actually hold the same value. Whereas, a verdict that
two pointers cannot may-alias must imply that there is no execution in which they hold the same
value.

We use a sound may-alias analysis in an attempt to prove the safety of a program with
respect to null-pointer exceptions. For each pointer dereference, we ask the analysis if
the pointer can may-alias \Null just before the dereference. If the answer is that it cannot, then
the pointer cannot hold a \Null value under all possible executions, hence the dereference is safe.
The more precise the analysis, the more
dereferences it can prove safe. This paper demonstrates a technique that improves the precision
of may-alias analysis at little cost when answering aliasing queries of pointers with the
\Null value.

The \Null value is special because programmers tend to be defensive against null-pointer exceptions.
If there is doubt that a pointer, say $\m{x}$, can be \Null or not, the programmer would use a check
``$\m{if \;(x \neq \Null)}$'' before dereferencing $\m{x}$. Existing alias analysis techniques, especially
flow insensitive techniques for may-alias analysis, ignore all branch conditions. 
As we demonstrate in this paper, exploiting these defensive checks can significantly increase the precision of 
alias analysis. Our technique is based around a semantics-preserving program transformation and requires only a minor change to the
alias analysis algorithm itself.

Program transformations have been used previously to improve the precision for alias
analysis. For instance, it is common to use a \textit{Single Static Assignment} (SSA) 
conversion \cite{CytronSSA} before running flow-insensitive analyses. The use of SSA
automatically adds some level of flow sensitivity to the analysis \cite{hasti1998using}.
SSA, while useful, is still limited in the amount of precision that it adds, and in particular,
it does not help with using branch conditions.
We present a program transformation based on \textit{Global Value Numbering} (GVN) \cite{KildallGVN}
that adds significantly more precision on top of SSA by leveraging branch conditions.

The transformation works by first inserting an assignment $\m{v := e}$ on the \textit{then} branch of
a check $\m{if \; (e \neq \Null)}$, where $\m{v}$ is a fresh program variable. This gives us the
global invariant that $\m{v}$ can never hold the \Null value. However, this invariant will be of no use unless
the program uses $\m{v}$. Our transformation then searches locally, in the same procedure, for 
program expressions $\m{e'}$ that are \textit{equivalent} to $\m{e}$, that is, at runtime they both hold the same value.
The transformation then replaces the use of $\m{e'}$ with $\m{v}$. The search for equivalent expressions
is done by adapting the GVN algorithm (originally designed for compiler optimizations \cite{DBLP:conf/popl/GulwaniN04}). 

Our transformation can be followed with a standard alias analysis 
to infer the points-to set for each variable, with a slight change that the new variables 
introduced by our transformation (such as $\m{v}$ above)
cannot be \Null. This change stops spurious propagation of \Null and makes the analysis
more precise. 
We perform extensive experiments on real-world code. The results show that 
the precision of the alias analysis (measured in terms of the number of 
pointer dereferences proved safe) goes from 86.56\% to 98.05\%. This work
is used inside Microsoft's Static Driver Verifier tool \cite{sdvurl} for
finding null-pointer bugs in device 
drivers\footnote{\url{https://msdn.microsoft.com/en-us/library/windows/hardware/mt779102(v=vs.85).aspx}}.

The rest of the paper is organized as follows: Section \ref{Se:Background} provides background on flow-insensitive 
alias analysis and how SSA can improve its precision. Section \ref{Se:Overview} illustrates our program transformation
via an example and Section \ref{Se:Algorithm} presents it formally.
Section \ref{Se:Experiments} presents experimental results,
Section \ref{Se:RelatedWork} describes some of the
related work in the area and Section \ref{Se:Conclusion} concludes.
Finally, Appendix \ref{Se:Proof} proves that the transformation preserves
program semantics.

\section{Background}
\label{Se:Background}

\subsection{Programming Language}
We introduce a simplistic language to demonstrate the alias analysis
and how program transformations can be used to increase its precision. 
As is standard, we concern ourselves only
with statements that manipulate pointers. All other statements are
ignored (i.e., abstracted away) by the alias analysis.
Our language has assignments with one of the following forms:
pointer assignments $\m{x := y}$, dereferencing via field writes $\m{x.f := y}$ and
reads $\m{x := y.f}$, creating new memory locations $\m{x := new()}$, or 
assigning \Null as $\m{x := \Null}$. The language also has
$\m{assume}$ and $\m{assert}$ statements:
\begin{itemize}
\item $\m{assume \; B}$ checks the Boolean condition $\m{B}$ and continues execution only if the condition
evaluates to \textit{true}. The $\m{assume}$ statement is a convenient way of modeling branching in most existing
source languages. For instance, a branch ``$\m{if \; (B)}$" can be modeled using two basic blocks, one
beginning with $\m{assume \; B}$ and the other with $\m{assume \; \neg B}$.
\item $\m{assert \; B}$ checks the Boolean condition $\m{B}$ and
continues execution if it holds. If $\m{B}$ does not hold, then it raises an assertion failure  and stops
program execution. 
\end{itemize}

A program in our language begins with global variable declarations
followed by one or more procedures.  Each procedure starts with 
declarations of local variables, followed by a sequence of basic blocks. Each basic block starts with a label,
followed by a list of statements, and ends with a control transfer, which is either a $\m{goto}$ or a
$\m{return}$. A $\m{goto}$ in our language can take multiple labels. The choice between which label to jump is
non-deterministic. Finally, we disallow loops in the control-flow of a procedure; they can instead be encoded using
procedures with recursion. This restriction simplifies the presentation of our algorithms. 
Figure~\ref{fig:example} shows an illustrative example in our language. 

\begin{figure}
\begin{minipage}{.5\linewidth}
\begin{lstlisting}[language=boogie, columns=fullflexible] 
var x : int
procedure f(var y : int) returns u : int
{
    var z : int
    L1:
        x := y.f;
        assume (x != Null);
        goto L2;
    L2:
	z.g := y;
	assert (x != Null);
	u := x;
	return;
}
\end{lstlisting}
\end{minipage}
\begin{minipage}{.5\linewidth}
\begin{lstlisting}[language=boogie, columns=fullflexible]
procedure main()
{
    var a : int;
    var b : int;
    L1:
    	a := new();
	b := call f(a);
	goto L2;
    L2:
	return;
}
\end{lstlisting}
\end{minipage}
\vspace{-3ex}
\caption{An example program in our language}
\vspace{-2ex}
\label{fig:example}
\end{figure}

%

\subsection{Alias Analysis}\label{sec:aa}
This section describes Andersen's may-alias analysis~\cite{andersen-thesis}.
The analysis is context and flow-insensitive, which means that it completely abstracts away
control of the program. But the analysis is field-sensitive, which means that a value can be obtained
by reading a field $\m{f}$ only if it was previously written to the same field $\m{f}$. Field-insensitive 
analyses, for example, also abstract away the field name.

The analysis outputs an over-approximation of  the set of memory locations each pointer can
hold under all possible executions.
Since a program can potentially execute indefinitely (because of loops or recursion), the number of 
memory locations allocated by a program can be
unbounded. We consider a finite abstraction of memory locations, commonly
called the \textit{allocation-site abstraction} \cite{jones1982flexible}. Each memory location
allocated by the same $\m{new}$ statement is represented using the same abstract value. An abstract
value is also called an allocation site. 
We label each $\m{new}$ statement with a unique number $i$ and refer to its corresponding allocation
site as $aS_i$. We use the special allocation site $\asN$ to denote $\Null$.

\begin{figure}
\begin{minipage}{.45\linewidth}
\begin{tabular}{c | c}
\textbf{Statement} & \textbf{Constraint} \\
\hline \\
[-9pt]
$\m{i: \quad x := new()}$ & $aS_i \in \pt{x}$ \\
[5pt]
$\m{x := \Null}$ & $aS_0 \in \pt{x}$ \\
[5pt]
$\m{x := y}$ & $\pt{y} \subseteq \pt{x}$ \\
[5pt]
$\m{x := y.f}$ & $\infer{\ptf{aS_i}{f} \subseteq \pt{x}}{aS_i \in \pt{y}}$ \\
[5pt]
$\m{x.f := y}$ & $\infer{\pt{y} \subseteq \ptf{aS_i}{f}}{aS_i \in \pt{x}}$
\end{tabular}
\caption{Program statements and \newline corresponding points-to set constraints}
\vspace{20ex}	
\label{constraints}
\end{minipage}
\begin{minipage}{.55\linewidth}
\begin{algorithm}[H]
\caption{Algorithm for computing points-to sets}
\label{aa-algo}
\begin{algorithmic}[1]
\State For each program variable $\m{x}$, let $\pt{x} = \emptyset$ 
\Repeat
  \State $opt  := pt$ 
  \ForAll{program statements $\texttt{st}$}
     \If{$\texttt{st}$ is $\m{i: x := new()}$}
         \State $\pt{x} := \pt{x} \cup \{aS_i\}$
     \EndIf
     \If{$\texttt{st}$ is $\m{x := \Null}$}
         \State $\pt{x} := \pt{x} \cup \{aS_0\}$
     \EndIf
     \If{$\texttt{st}$ is $\m{x := y}$}
         \State $\pt{x} := \pt{x} \cup \pt{y}$
     \EndIf
     \If{$\texttt{st}$ is $\m{x := y.f}$}
         \ForAll{$aS_i \in \pt{y}$}
            \State $\pt{x} := \pt{x} \cup \ptf{aS_i}{f}$
         \EndFor
     \EndIf
     \If{$\texttt{st}$ is $\m{x.f := y}$}
         \ForAll{$aS_i \in \pt{x}$}
            \State $\ptf{aS_i}{f} := \ptf{aS_i}{f} \cup \pt{y}$
         \EndFor
     \EndIf
     \ForAll{$\tagged(\m{x})$} \label{tagged-line}
        \State $\pt{x} := \pt{x} - \{aS_0\}$
      \EndFor
  \EndFor
\Until{$opt = pt$}
\end{algorithmic}

\end{algorithm}
\end{minipage}
\end{figure}

We follow a description of Andersen's analysis in terms of set constraints \cite{manu-lncs},
shown in Figure~\ref{constraints}.
The analysis outputs a points-to relation 
$pt$ where $\pt{x}$ represents the points-to set of a variable $\m{x}$, 
i.e. (an over-approximation of) the set of allocation sites that $\m{x}$ may hold under all possible executions. In addition, it
also computes $\ptf{aS_i}{f}$, for each allocation site $aS_i$ and field $\m{f}$, representing (an over-approximation of)
the set of values written to the $\m{f}$ field of an object represented by $aS_i$.

The analysis abstracts away program control along with $\m{assert}$ and
$\m{assume}$ statements. It essentially considers a program as a bag
of pointer-manipulating statements where
any statement can be executed any number of times and in any order.
For each statement, the analysis follows Figure~\ref{constraints} to generate
a set of rules that define constraints on the points-to solution $pt$. The rules can be read as follows.

\begin{itemize}
\item If a program has an allocation $\m{x := new()}$ and this statement is labelled with the unique integer $i$, then
the solution must have $aS_i \in \pt{x}$.
\item If a program has the statement $\m{x := NULL}$, then it must be that $aS_0 \in \pt{x}$.
\item If the program has an assignment $\m{x := y}$ then the solution must have $\pt{y} \subseteq \pt{x}$, 
because $\m{x}$ may hold any value that $\m{y}$ can hold. 
\item If the program has a statement $\m{x := y.f}$ and $aS_i \in \pt{y}$, then it must be that
$\ptf{aS_i}{f} \subseteq \pt{x}$ because $\m{x}$ may hold any value written to the $\m{f}$ field of $aS_i$.
\item If the program has a statement $\m{x.f := y}$ and $aS_i \in \pt{x}$ then it must be that
$\pt{y} \subseteq \ptf{aS_i}{f}$.
\end{itemize}

These set constraints can be solved using a simple fix-point iteration, shown in Algorithm \ref{aa-algo}.
(Our tool uses a more efficient implementation \cite{manu-lncs}.) For now, ignore the loop on line \ref{tagged-line}.
Once the solution is computed,
we check all assertions in the program. We say that an assertion $\m{assert \; (x \neq \Null)}$
is \textit{safe} (i.e., the assertion cannot be violated) if $\asN \not\in \pt{x}$. We do not consider
other kinds of assertions in the program because our goal is just to show null-exception safety.
Andersen's analysis complexity is cubic in the size of the program,
and for k-sparse programs, the worst case complexity is quadratic~\cite{SridharanF09}.

\subsection{Static Single Assignment (SSA)}

This section shows how a program transformation can improve the precision of an alias analysis. 
Consider the following program.
\begin{boogie}
x := new();
assert (x != Null);
y := x.f;
x := Null;
\end{boogie}
A flow-insensitive analysis does not look at the order of statements. Under this abstraction,
the analysis cannot prove the safety of the assertion in the above snippet of code because it does
not know that the assignment of $\Null$ to $\m{x}$ only happens after the assertion.

To avoid such loss in precision, most practical implementations of alias analysis use the Single Static
Assignment (SSA) form \cite{CytronSSA}. Roughly, SSA introduces multiple copies of each original variable
such that each variable in the new program only has a single assignment. The following is the SSA form
of the snippet shown in the beginning of this section.

\begin{boogie}
x1 := new();
assert (x1 != Null);
y := x1.f;
x2 := Null;
\end{boogie}

Clearly, this program has the same semantics as the original program. But a flow-insensitive analysis
will now be able to show the safety of the assertion in the program because the assignment of $\Null$
is to $\m{x2}$ whereas the assertion is on $\m{x1}$.

\section{Overview}
\label{Se:Overview}

This section presents an overview of our technique of using stronger
program transformations that add even more precision to the alias analysis
compared to the standard SSA. We start by using 
\textit{Common Subexpression Elimination} \cite{CockeCSE} and build towards
using \textit{Global Value Numbering} \cite{KildallGVN}, which is used in our 
implementation and experiments.

\subsection{Common Subexpression Elimination}
We demonstrate how we can leverage $\m{assume}$ and $\m{assert}$ statements
to add precision to the analysis. Consider the following program.
\begin{boogie}
assume (x != Null);
y := x;
assert (x != Null);
z := x.f;
\end{boogie}
Once the program control passes the $\m{assume}$ statement, we know that
$\m{x}$ cannot point to $\m{\Null}$, hence the assertion is safe, irrespective of what preceded this code
snippet. Note that SSA renaming does not help prove the assertion in this case (it is essentially a no-op for
the above snippet). We now make the case for a different program transformation.

As a first step, we introduce a new local variable $\m{cseTmp^\varsp}$ to the procedure and assign it the value
of $\m{x}$ right after the $\m{assume}$. These new variables that we introduce to the program will carry
the tag ``$\varsp$'' to distinguish them from other program variables. For a tagged variable $\m{w^\varsp}$, we say that
$\tagged(\m{w}^\varsp)$ is \textit{true}. These tagged variables carry the special invariant that they cannot be 
\Null; their only assignment will be after an assume statement that prunes away the \Null value.

After introducing the variable $\m{cseTmp^\varsp}$, we make use of a technique similar to
\textit{Common Subexpression Elimination} (CSE) to replace all expressions
that are equivalent to $\m{cseTmp^\varsp}$ with the variable itself, resulting in the following program:

\begin{boogie}
assume (x != Null);
cseTmp$^\varsp$ := x;
y := cseTmp$^\varsp$;
assert (cseTmp$^\varsp$ != Null);
z := cseTmp$^\varsp$.g;
\end{boogie}

This snippet is clearly equivalent to the original one. We perform the alias analysis on this
snippet as usual, but enforce that $\pt{cseTmp^\varsp}$ cannot have $\asN$ because it cannot be $\Null$.
(See the loop on line \ref{tagged-line} of Algorithm \ref{aa-algo}.)
The analysis can now prove that the assertion is safe.

The process of finding equivalent expressions is not trivial. For instance, consider the following program
where we have introduced the variable $\m{cseTmp^\varsp}$.
\begin{boogie}
assume (x.f != Null);
cseTmp$^\varsp$ := x.f;
y.f := z;
z := x.f;
\end{boogie}
In the last assignment, $\m{x.f}$ cannot be substituted by $\m{cseTmp^\varsp}$,
because there is an assignment to the field $\m{f}$ in the previous statement.
As there is no aliasing information present at this point, we have to conservatively
assume that $\m{y}$ and $\m{x}$ could be aliases, thus, the assignment $\m{y.f := z}$
can potentially change the value of $\m{x.f}$, breaking its equivalence to $\m{cseTmp^\varsp}$.

\subsection{Global Value Numbering} \label{sec:gvn}
We improve upon the previous transformation by using a stronger method of determining expression
equalities. The methodology remains the same: we introduce
temporary variables that cannot be $\m{\Null}$ and use them to replace syntactically equivalent expressions.
But this time we adapt the Global Value Numbering (GVN) scheme to detect equivalent expressions.
Consider the following program. (For now, ignore the right-hand side of the figure after the
$\Longrightarrow$.)
\begin{boogie}[mathescape=true, numbers=left, xleftmargin=14pt,xrightmargin=14pt, firstnumber=1]
y := x.f.g;             ==> $t_1 \gets \m{x}, \; t_2 \gets t_1\m{.f}, \; t_3 \gets t_2\m{.g}, \; \m{y} \gets t_3$
z := y.h;               ==> $t_3 \gets \m{y}, \; t_4 \gets t_3\m{.h}, \; \m{z} \gets t_4$
assume (z != Null);      ==> set $t_4 \in \{\nonNullExprs\}$
a := x.f;               ==> $t_1 \gets \m{x}, \; t_2 \gets t_1\m{.f}, \; \m{a} \gets t_2$
b := a.g.h;             ==> $t_2 \gets \m{a}, \; t_3 \gets t_2\m{.g}, \; t_4 \gets t_3\m{.h}, \; \m{b} \gets t_4$
assert (b != Null);      ==> check $t_4 \in \{\nonNullExprs\}$
c.g := d;
\end{boogie}
It is clear that $\m{z}$ and $\m{b}$ are equivalent at the assertion location,
and since $\m{z \neq \Null}$, the assertion is safe. However, none of the previous methods would allow
us to prove the safety of the assertion. We adapt the GVN scheme to help us
establish the equality between $\m{z}$ and $\m{b}$.
We introduce the concept of \textit{terms} that will be used as a
placeholder for subexpressions. The intuitive idea is that equivalent
subexpressions will be represented using the same term. 

We start by giving an overview of the transformation for a single basic block,
and then generalize it to full procedure later in this section.
For a single basic block, we walk through the statements in order and as we encounter
a new variable, we assign it a new term and remember this mapping in a dictionary
called $hashValue$. We also store the mapping from terms to other terms
through operators in a separate dictionary called $\hashFunction$. For example, if $\m{x}$
is assigned term $t_1$, and we encounter the assignment $\m{y := x.f}$, we store
$\hashFunction[\m{f}][t_1] = t_2$ and assign the term $t_2$ to $\m{y}$.
We also maintain a separate list $\nonNullExprs$ of
terms that are not null. Finally, for performing the actual substitution, we maintain
a dictionary $\defaultVar$ that maps terms to the temporary variables that we introduce
for non-null expressions.

We go through the program snippet starting at the first statement and move down to the last statement. At
statement $i$, we follow the description written in the $i^\text{th}$ item below. This description is also
shown on the right side of the program snippet, after the $\Longrightarrow$ arrow.
\begin{sloppy}
\begin{enumerate}
\item Assign a new term $t_1$ to $\m{x}$, and set $hashValue[\m{x}] = t_1$. Then,
set $\hashFunction[\m{f}][t_1] = t_2$, and $\hashFunction[\m{g}][t_2] = t_3$.
Finally the assignment to $\m{y}$ adds $hashValue[\m{y}] = t_3$.

\item We already have $hashValue[\m{y}] = t_3$, so assign
$\hashFunction[\m{g}][t_3] = t_4$. The assignment adds $hashValue[\m{z}] = t_4$.

\item We have $hashValue[\m{z}] = t_4$. So, we add $t_4$ to $\nonNullExprs$.
We create a new temporary variable $\m{gvnTmp}^\varsp$, and construct
an extra assignment $\m{gvnTmp^\varsp := z}$, and add it after the $\m{assume}$
statement. Because $hashValue[\m{z}] = t_4$, we also add 
$\defaultVar[t_4] = \m{gvnTmp^\varsp}$, which we will use later for substitutions
to all expressions that hash to $t_4$.

\item We already have $hashValue[\m{x}] = t_1$ and $\hashFunction[\m{f}][t_1] = t_2$,
so we add $hashValue[\m{a}] = t_2$.

\item We have $hashValue[\m{a}] = t_2$, $\hashFunction[\m{g}][t_2] = t_3$ and
$\hashFunction[\m{h}][t_3] = t_4$. Hence, the hash value of the expression $\m{a.g.h}$ is $t_4$.
We also have $\defaultVar[t_4] = \m{gvnTmp^\varsp}$. At this point, we observe
$t_4$ being in $\nonNullExprs$ and substitute
the RHS $\m{a.g.h}$ with $\m{gvnTmp^\varsp}$. Finally, we add $hashValue[\m{b}] = t_4$.

\item Because $hashValue[\m{b}] = t_4$ and $\defaultVar[t_4] = \m{gvnTmp^\varsp}$
and $\nonNullExprs$ contains $t_4$,
we replace the expression $\m{b}$ with $\m{gvnTmp^\varsp}$. 
\end{enumerate}
\end{sloppy}

The resulting code is shown below. 
\begin{boogie}[mathescape=true, numbers=left, xleftmargin=14pt,xrightmargin=14pt, firstnumber=1]
y := x.f.g;             
z := y.h;               
assume (z != Null);     
gvnTmp$^\varsp$ := z;
a := x.f;               
b := gvnTmp$^\varsp$;   
assert (gvnTmp$^\varsp$ != Null);
c.g := d;
\end{boogie}

Clearly, we retain the invariant that $\#$-tagged variables
cannot be \Null, and it is now straightforward to
prove the safety of the assertion. We also note that the expression
substitution is performed in a conservative manner. It
is aborted as soon as a subexpression is assigned to. For example, at
line 8, we encounter an assignment to the field $\m{g}$, so we remove
$\m{g}$ from the dictionary $\hashFunction$. This has the effect of $\m{g}$
acting as a new field, and all terms referenced by this field will now be assigned
new terms.

The above transformation, in general, is performed on the entire procedure,
not just a basic block to fully exploit its potential. This occurs in two steps.
First, loops are lifted and converted to procedures (with recursion), so that the control-flow
of each resulting procedure is acyclic. Next, we perform a topological sort~\cite{aspvall1979linear}
of the basic blocks of a procedure and analyze the blocks in this order. This ensures that by the
time the algorithm visits a basic block, it has already processed all predecessors of the block.

When analyzing a block, the algorithm considers all its predecessors and takes
the intersection of their $\nonNullExprs$ list and $hashValue$ map. This is because
an expression is non-null only if it is non-null in all its predecessors and, further, we can
use a term for a variable only if it is associated with the same term in all its
predecessors.
Finally, an important aspect of the algorithm is to perform a sound substitution at
the merge point of two basic blocks. Consider the code snippet below.
\begin{boogie}
L1:
  assume (x != Null);
  gvnTmp$^\varsp_1$ := x;
  goto L3;
L2:
  assume (x != Null);
  gvnTmp$^\varsp_2$ := x;
  goto L3;
L3:
  assert (x != Null);
\end{boogie}
In this example, although $\m{x}$ is available as a non-null expression in $\m{L3}$,
we cannot substitute $\m{x}$ in the assertion by either
$\m{gvnTmp}^\varsp_1$ or $\m{gvnTmp}^\varsp_2$ because neither preserves program
semantics. Instead, we introduce a
new variable $\m{gvnTmp}^\varsp_3$ and add the assignment
$\m{gvnTmp}^\varsp_3 := \m{x}$ right before the assertion in $\m{L3}$ and use 
that for substituting $\m{x}$.
This is achieved by the map $var2expr$
in the main algorithm. It maps the current block and the $\varsp$-tagged
variable to the expression it will substitute for. In the above program, let's say we
assign the term $t$ to the non-null expression $\m{x}$. Hence, $\nonNullExprs[\m{L1}]$
and $\nonNullExprs[\m{L2}]$ both contain $t$. We also have 
$\defaultVar[\m{L1}][t] = \m{gvnTmp}^\varsp_1$ and $var2expr[\m{gvnTmp}^\varsp_1] = \m{x}$.
Since $t$ is available
from all predecessors of $\m{L3}$, we know that this term is non-null in $\m{L3}$.
The question is finding the expression corresponding to this term and introducing
a new assignment for it. At this point, the map $var2expr$ comes into play. We pick
a predecessor of $\m{L3}$, say $\m{L1}$. We look for the default variable of $t$
and find $\defaultVar[\m{L1}][t] = \m{gvnTmp}^\varsp_1$, we then search for
$var2expr[\m{gvnTmp}^\varsp_1] = \m{x}$. At this point, we find that the expression
corresponding to term $t$ is $\m{x}$, and we introduce a new assignment
$\m{gvnTmp}^\varsp_3 := \m{x}$ at the start of $\m{L3}$ and use this for
substitution of $\m{x}$.
With these motivating examples, the next section describes the algorithm formally.

\section{Algorithm}
\label{Se:Algorithm}

We present the pseudocode of our program transformation in this section (Algorithms~\ref{gvn-algo}
and \ref{helper}). The transformation takes a
program as input and produces a semantically equivalent program with
new $\varsp$-tagged variables that can never be $\Null$. This involves
adding assignments for these new variables, and substituting existing
expressions with these variables whenever we determine that the substitution
will preserve semantics.

\begin{algorithm}[H]
\caption{Algorithm to perform GVN}
\label{gvn-algo}
\begin{algorithmic}[1]
\State $\nonNullExprs = \emp$ \Comment{block $\rightarrow$ non-null terms in block}
\State $var2expr = \emp$ \Comment{$\varsp$-tagged variable $\rightarrow$ expression}
\State $\defaultVar = \emp$ \Comment{block, term $\rightarrow$ variable for substitution}
\State $hashValue = \emp$ \Comment{block, variable $\rightarrow$ term}
\State $\hashFunction = \emp$ \Comment{operator, terms $\rightarrow$ term}
\State $currBlock$ \Comment{current block}
\Function{DoGVN}{}
\For{$proc$ in $program$} \label{start-add-asgn}
	\For{$block$ in $proc.Blocks$}
		\For{$stmt$ in $block.Stmts$}
			\If{$stmt$ is ``$\m{assert \; expr \neq \Null}$" or ``$\m{assume \; expr \neq \Null}$"}
				\State $\m{gvnTmp^\varsp} \gets GetNewSpecialVar()$
				\State $s \gets ``\m{gvnTmp^\varsp := expr}"$ \label{asgn1}
				\State $block.Stmts.Add(s)$
				\State $var2expr[block][\m{gvnTmp^\varsp}] \gets \m{expr}$
			\EndIf
		\EndFor
	\EndFor
\EndFor \label{end-add-asgn}
\For{$proc$ in $program$}
	\State $sortedBlocks \gets TopologicalSort(proc.Blocks)$ \label{topsort}
	\For{$block$ in $sortedBlocks$} \label{start-pre-blk}
		\State $\nonNullExprs[block] \gets \bigcap_{blk \in block.Preds} \nonNullExprs[blk]$ \label{lem2-proof}
		\State $hashValue[block] \gets \bigcap_{blk \in block.Preds} hashValue[blk]$ \label{lem3-proof}
		\State $currBlock \gets block$
		
		\For{$term$ in $\nonNullExprs[block]$}
			\State $\m{expr} \gets var2expr[\defaultVar[blk][term]]$ \Comment{for some $blk \in block.Preds$}
			\State $\m{gvnTmp^\varsp} \gets GetNewSpecialVar()$
			\State $var2expr[\m{gvnTmp^\varsp}] \gets \m{expr}$
			\State $s \gets ``\m{gvnTmp^\varsp} := \m{expr}"$ \label{asgn2}
			\State $block.Stmts.Add(s)$
		\EndFor
	\EndFor \label{end-pre-blk}
	\For{$stmt$ in $block.Stmts$} \label{start-proc}
		\State $stmt \gets ProcessStmt(stmt)$
		\If{$stmt$ is $``\m{gvnTmp^\varsp} := \m{expr}"$}
			\State $term \gets ComputeHash(\m{expr})$
			\State $\nonNullExprs[block].Add(term)$
			\State $\defaultVar[block][term] \gets \m{gvnTmp^\varsp}$ \label{def-asgn}
		\EndIf
	\EndFor \label{end-proc}
\EndFor
\EndFunction
\end{algorithmic}
\end{algorithm}

At a high level, the idea is to use $\m{assume}$ and $\m{assert}$ statements to identify 
non-null expressions. We introduce fresh $\varsp$-tagged
variables and assign these non-null expressions
to them. Then, in a second pass, we compute a \textit{term} corresponding
to each expression. These terms are assigned in a manner that
if two expressions have the same term, then they are equivalent to each other. 
If we encounter an expression $e$ with the same term as one of the
non-null expressions $e'$, we substitute $e$ with the $\varsp$-tagged variable
corresponding to $e'$.

We start by describing
the role of each data structure used in Algorithm~\ref{gvn-algo}.
\begin{itemize}
\item $\nonNullExprs$: For each block, this stores the terms of non-null expressions
for that block.

\item $var2expr$: This maps a $\varsp$-tagged variable to the
expression it is assigned to in each block. This will be used to solve the issue
discussed in the last example of Section~\ref{sec:gvn}.

\item $\defaultVar$: This maps the term corresponding to an
expression to the $\varsp$-tagged variable that will be used for its substitution.
Whenever we compute the term for an expression, if the term is present
in $\nonNullExprs$, we will use $\defaultVar$ to find the $\varsp$-tagged variable that is
going to be used for the substitution.

\item $hashValue$: It stores the term assigned to each variable in
a particular block.

\item $\hashFunction$: It stores the mapping from a field and a term
to a new term. It is used to store the term for expressions with fields.

\item $currBlock$: It keeps track of the current block and
is used while calling the helper functions.
\end{itemize}

We now explain the algorithm step by step.

\begin{enumerate}
\item 
Lines~\ref{start-add-asgn} - \ref{end-add-asgn} -- In this first pass of
the algorithm, we search for program statements of the form
``$\m{assert \; expr \neq \Null}$" or ``$\m{assume \; expr \neq \Null}$".
This guarantees that $\m{expr}$ cannot be $\Null$ after this program
location under all executions. Hence, we introduce a new variable
$\m{gvnTmp}^\varsp$ and assign $\m{expr}$ to it. This mapping is also added
to $var2expr$.

\item 
Line~\ref{topsort} -- Before doing the second pass, we perform a topological
sort on the blocks according to the control flow graph. This is necessary since
we need the information of $\nonNullExprs$ for the predecessors of a basic
block before analyzing it. Note that conflow-flow graphs of procedures in our language
must be acyclic (we convert loops to recursion), thus a topological sorting always
succeeds.

\item 
Lines~\ref{start-pre-blk} - \ref{end-pre-blk} -- We compute the set of
expressions that are non-null in all predecessors. Only these expressions
will be non-null in the current block. We also need the term for
each variable in the current block, which also comes from the intersection
of terms from all predecessors. Finally, for all the non-null expressions,
we add an assignment since these expressions may be available from
different variables in different predecessors, as discussed in Section~\ref{sec:gvn}.

\item 
Lines~\ref{start-proc} - \ref{end-proc} -- Finally, we process each statement
in the current block. This performs the substitution for each expression in the statement
($GetExpr$ function in Algorithm~\ref{helper}). $GetExpr$ computes the term
for the expression ($ComputeHash$ function in Algorithm~\ref{helper}), and
if the term is contained in $\nonNullExprs$, the substitution is performed.
Finally, if we encounter a store statement, ``$\m{v.f := expr}$", we remove
all mappings w.r.t. $\m{f}$ in $\hashFunction$. So, for the future statements
(and future blocks in the topological order), new terms will be assigned
to expressions related to field $\m{f}$.
\end{enumerate}

With this pseudocode, we will generate a semantically equivalent program,
and as we show in our experiments, will have improved precision with regard
to alias analysis. The main reason behind this improvement is that
these $\varsp$-tagged variables can never contain $\asN$ in the points-to set, hence
$\asN$ cannot flow through these variables in the analysis, while
earlier, there was no such restriction and $\Null$ could flow freely.
The pseudocode for the algorithm is demonstrated in Algorithms~\ref{gvn-algo}
and \ref{helper}. 

\begin{algorithm}[H]
\caption{Helper Functions for DoGVN}\label{helper}
\begin{algorithmic}[1]
\Function{ProcessStmt}{$stmt$}
\If{$stmt$ is ``$\m{assume \; expr}$" or ``$\m{assert \; expr}$"}
	\State $\m{expr} \gets GetExpr(\m{expr})$
	\State \Return $stmt$
\ElsIf{$stmt$ is ``$\m{v := expr}$"}
	\State $hashValue[currBlock][\m{v}] \gets ComputeHash(\m{expr})$
	\State $\m{expr} \gets GetExpr(\m{expr})$
	\State \Return $stmt$
\ElsIf{$stmt$ is ``$\m{v.f := expr}$"}
	\State $\m{expr} \gets GetExpr(\m{expr})$
	\State $\m{v} \gets GetExpr(\m{v})$
	\State $\hashFunction.Remove(\m{f})$
	\State \Return $stmt$
\EndIf
\EndFunction
\Function{GetExpr}{$\m{expr}$}
\If{$\m{expr}$ is $\m{v}$}
	\State $term \gets ComputeHash(\m{v})$
	\If{$\nonNullExprs[currBlock]$ contains $term$} \label{lem1-proof}
		\State \Return $\defaultVar[currBlock][term]$ \label{main-proof1}
	\EndIf
	\State \Return $\m{v}$
\EndIf
\If{$\m{expr}$ is ``$\m{v.f}$"}
	\State $\m{v} \gets GetExpr(\m{v})$
	\State \Return ``$\m{v.f}$"
\EndIf
\EndFunction
\Function{ComputeHash}{$\m{expr}$}
\If{$\m{expr}$ is $\m{v}$}
	\If{$hashValue[currBlock]$ does not contain $\m{v}$}
		\State $hashValue[currBlock][\m{v}] \gets GetNewTerm()$
	\EndIf
	\State \Return $hashValue[currBlock][\m{v}]$
\ElsIf{$\m{expr}$ is ``$\m{v.f}$"}
	\State $term \gets ComputeHash(\m{v})$
	\If{$\hashFunction[\m{f}]$ does not contain $term$}
		\State $\hashFunction[\m{f}][term] \gets GetNewTerm()$ \label{hf-update}
	\EndIf
	\State \Return $\hashFunction[\m{f}][term]$
\EndIf
\EndFunction
\end{algorithmic}
\end{algorithm}

%

\section{Experimental Evaluation}
\label{Se:Experiments}

We have implemented the algorithms presented in this paper
for the Boogie language  \cite{Leino08boogie}. Boogie
is an intermediate verification language. 
Several front-ends are available that compile source languages 
(such as C/C++ \cite{conf/cav/RakamaricE14,DBLP:conf/sigsoft/LalQ14} 
and C\# \cite{BCT}) to Boogie, making it a useful target for
developing practical tools. 

Our work fits into a broader verification effort.
The \textit{Angelic Verification} (AV) project\footnote{\url{https://www.microsoft.com/en-us/research/project/angelic-verification/}}
at Microsoft Research aims to design push-button technology for finding software defects.
In an earlier effort, AV was targeted to find null-pointer bugs \cite{das2015angelic}. 
Programs from the Windows codebase, in C/C++, were compiled down to Boogie with assertions guarding
every pointer access to check for null dereferences. These Boogie programs were fed to a verification
pipeline that applied heavyweight SMT-solver  technology to reason over all possible program behaviors.
To optimize the verification time, an alias analysis is run at the Boogie level to remove
assertions that can be proved safe by the analysis. As our results will show, this optimization is necessary. The alias
analysis is based on Andersen's analysis, as was described in Figure~\ref{constraints}. We follow
the algorithm given in Sridharan et al.'s report~\cite{manu-lncs} with
the extra constraint that $\#$-tagged variables cannot alias
with \Null, i.e. they cannot contain the allocation site $\asN$. We can optionally perform the program
transformation of Section~\ref{Se:Algorithm} before running the alias analysis. 
Our implementation is available open-source\footnote
{At \url{https://github.com/boogie-org/corral}, project
\textbf{AddOns$\backslash$AliasAnalysis}}.


\begin{table}[t]
\setlength{\tabcolsep}{5pt}
\centering
\begin{tabular}{| l | r  r  r | r r | r r r |}
\hline
& \multicolumn{3}{|c}{\textbf{Stats}} & \multicolumn{2}{| c }{\textbf{SSA only}} & \multicolumn{3}{| c |}{\textbf{SSA with GVN}} \\
\textbf{Bench} & \textbf{Procs} & \textbf{KLOC} & \textbf{Asserts} & \textbf{Time(s)} & \textbf{Asserts} & \textbf{Time(s)} & \textbf{GVN(s)} & \textbf{Asserts} \\
\hline
\powerfx & 40 & 3.2 & 1741 & 9.08 & 61 & 11.37 & 0.88 & 17 \\
\diskdump & 37 & 8.4 & 4035 & 11.34 & 233 & 17.62 & 1.13 & 45 \\
\hid & 64 & 6.5 & 4375 & 10.26 & 617 & 19.43 & 2.15 & 52 \\
\usbstor & 184 & 20.9 & 7523 & 24.04 & 543 & 33.99 & 2.43 & 123 \\
\mpiomsdm & 284 & 30.9 & 11184 & 35.02 & 1881 & 59.84 & 7.11 & 232 \\
\mpiocontrol & 382 & 37.8 & 12128 & 35.94 & 2675 & 70.71 & 11.13 & 452 \\
\gpio & 453 & 37.2 & 6840 & 36.88 & 1396 & 53.24 & 3.44 & 127 \\
\scsiport & 400 & 43.8 & 12209 & 28.91 & 2854 & 62.27 & 5.38 & 475 \\
\iscsi & 479 & 56.6 & 19030 & 60.05 & 5444 & 106.61 & 12.40 & 508 \\
\usbhub & 998 & 76.5 & 39955 & 171.43 & 2887 & 839.58 & 475.08 & 372 \\
\xusb & 867 & 23.5 & 6966 & 49.17 & 875 & 69.10 & 10.14 & 103 \\
\db & 303 & 14.9 & 8359 & 24.57 & 820 & 59.13 & 13.41 & 210 \\
\io & 419 & 22.1 & 11471 & 38.27 & 869 & 87.07 & 24.03 & 248 \\
\port & 493 & 36.2 & 18026 & 48.56 & 2501 & 149.60 & 41.93 & 478 \\
\ksfilter & 317 & 19.4 & 20555 & 55.07 & 586 & 269.35 & 134.06 & 131 \\
\audio & 809 & 54.0 & 16957 & 62.86 & 2821 & 127.67 & 30.46 & 342 \\
\hline
\textbf{Total} & \textbf{6529} & \textbf{491.9} & \textbf{201354} & \textbf{701.45} & \textbf{27063} & \textbf{2036.58} & \textbf{775.16} & \textbf{3915} \\
\hline
\end{tabular}
\caption{Results showing the effect of SSA and GVN program transformations on the ability of alias analysis to prove safety of non-null assertions.}
\vspace{-2ex}
\label{restable}
\end{table}

We evaluate the effect of our program transformation 
on the precision of alias analysis for checking safety of null-pointer
assertions. The benchmarks are described in the first four
columns of Table~\ref{restable}. We picked $16$ different modules
from the Windows codebase. The table lists an anonymized name for the 
module (\textbf{Bench}), the number of procedures contained in the module (\textbf{Procs}),
the lines of code in thousands (\textbf{KLOC}) and the number of assertions (one per pointer dereference)
in the code (\textbf{Asserts}). It is worth noting that the first ten modules
are the same as ones used in the study with AV \cite{das2015angelic}, while the rest
were added later.

We execute our tool in two modes. In the first, we use SSA and then run the alias 
analysis algorithm. In the second, we perform our GVN transformation on top of SSA
and then run the
alias analysis algorithm. In each case, we list the total time taken by the tool (\textbf{Time(s)}),
including the time to run the transformation, and the number of assertions that were
\textit{not} proved safe (\textbf{Asserts}). In the case of GVN, we also isolate
and list the time taken by the GVN transformation itself (\textbf{GVN(s)}).

The experiments were run (sequentially, single-threaded) on a
server class machine with an Intel(R) Xeon(R) processor (single core)
executing at 2.4 GHz with 32 GB RAM.

It is clear from the table that GVN offers significant increase in precision.
With only the use of SSA, the analysis was able to prove the safety of 
$86.56\%$ of assertions, while with the GVN transformation, we can prune away
$98.05\%$ of assertions. This is approximately a $7$X reduction in the number of
assertions that remain. This pruning is surprising because the alias analysis is still
context and flow insensitive. Our program transformation crucially exploits the fact
that programmers tend to be defensive against null-pointer bugs, allowing the analysis 
to get away with a very coarse abstraction. In fact, this level of pruning meant that
any level of investment in making the analysis more sophisticated (e.g., flow sensitive)
would have very diminished returns. 

The alias analysis itself scales quite well: it finishes on about half a million
lines of code in approximately $700$ seconds with just SSA ($86.56\%$ pruning)
or $2000$ seconds with GVN ($98.05\%$ pruning). We note that there is an increase
in the running time when using GVN. This happens because the transformation introduces
more variables, compared to just SSA. However, this increase in time is more than 
offset by the improvement presented to the AV toolchain. For example, with the GVN transformation,
AV takes $11$ hours to finish the first $10$ modules, whereas with the SSA transformation alone it
does not finish even in $24$ hours. Furthermore, AV reports fewer bugs when using just 
SSA because the extra computational efforts translate to a loss in program coverage as timeouts
are hit more frequently.

\section{Related Work}
\label{Se:RelatedWork}

Pointer analysis is a well-researched branch of static analysis. 
There are several techniques proposed that interplay between 
context, flow and field sensitivity. Our choice of using context-insensitive,
flow-insensitive but field sensitive analysis is to pick a scalable starting
point, after which we add precision at low cost. The distinguishing factor in our
work is: $(1)$ the ability to leverage information from $\m{assume}$ and $\m{assert}$ statements 
(or branch conditions) and $(2)$ specializing for the 
purpose of checking non-null assertions (as opposed to general aliasing assertions).
We very briefly list, in the rest of this section, some of the previous work in adding 
precision to alias analysis or making it more scalable. 

\paragraph{Context Sensitivity.} 
Sharir and Pnueli~\cite{sharir1978two} introduced the concept of \textit{call-strings}
to add context-sensitivity to static analysis techniques. Call strings may 
grow extremely long and limit efficiency, so Lhot\'{a}k and
Hendren~\cite{spark} used k-limiting approaches to limit the size
of call strings. Whaley and Lam~\cite{whaley-lam-lncs} instead use 
Binary Decision Diagrams (BDDs) to scale a context sensitive analysis.

\paragraph{Flow sensitivity.}
Hardekopf and Lin~\cite{HardekopfCGO} present a staged flow-sensitive
analysis where a less precise auxiliary pointer analysis computes def-use
chains which is used to enable the sparsity of the primary flow-sensitive
analysis. The technique is quite scalable on large benchmarks but
they abstract away the assume statements. De and D'Souza~\cite{DeECOOP}
compute a map from access paths to sets of abstract objects at each
program statement. This enables them to perform strong updates at
indirect assignments. The technique is shown to be scalable only for small
benchmarks, moreover, they also abstract away all assume statements.
Finally, Lerch et al.~\cite{LerchASE} introduce the access-path abstraction,
where access paths rooted at the same base variable are represented
by this base variable at control flow merge points. The technique is quite
expensive even on small benchmarks (less than 25 KLOC) and do not
deal with assume statements in any way.

\paragraph{Other techniques.}
Heintze and Tardieu~\cite{heintze2001demand} improved performance
by using
a demand-driven pointer analysis, computing sufficient information
to only determine points-to set of query variables.
Fink et al.~\cite{FinkTypestate} developed a staged
verification system, where faster and naive techniques run as early
stages to prune away assertions that are easier to prove, which then reduces
the load on more precise but slow techniques that run later. Landi and
Ryder~\cite{LandiConditional} use conditional may alias information
to over-approximate the points-to sets of each pointer. Context
sensitivity is added using k-limiting approach, and a set of aliases
is maintained for every statement within a procedure to achieve
flow-sensitivity. Choi et al.~\cite{ChoiCFS} also follows~\cite{LandiConditional}
closely but uses sparse representations for the control flow graphs
and use transfer functions instead of alias-generating rules.
To the best of our knowledge, none of these techniques are able to
leverage $\m{assume}$ statements to improve precision.

\section{Conclusion}
\label{Se:Conclusion}

This paper presents a program transformation that improves the efficiency
of alias analysis with minor scalability overhead. The transformation is proved 
to be semantics preserving. Our evaluation demonstrates the merit of our
approach on a practical end-to-end scenario of finding null-pointer dereferences
in software. 



\bibliographystyle{eptcs}
\bibliography{refs}

\appendix

\section{Proof of Correctness}
\label{Se:Proof}

We sketch the proof of the fact that our transformation preserves semantics.
To substitute an expression, say $\m{expr}$ with a variable, say 
$\m{v}$ at a program
location $L$, we need to prove the following two conditions.
\begin{itemize}
\item Assignment of $\m{v}$ reaches $L$ along
every execution path.
\item $\m{expr}$ and $\m{v}$ evaluate to the same value at $L$
under all possible executions.
\end{itemize}

\subsection{First Condition}
We begin by proving the first condition. Note that substitution only occurs
in the function $GetExpr()$, (line~\ref{main-proof1} in Algorithm~\ref{helper})
and only by variables present in the map
$defaultVar$. Also, only $\varsp$-tagged variables are added to
$defaultVar$ (line~\ref{def-asgn} in Algorithm~\ref{gvn-algo}). 
Hence, $\m{v}$ is tagged with $\varsp$. As is clear from the
algorithm, such a variable is assigned either at line~\ref{asgn1}
or line~\ref{asgn2} in Algorithm~\ref{gvn-algo}. Moreover, in both cases,
this variable is generated afresh before constructing the assignment. Hence,
$\m{v}$ is assigned only once. Let this assignment
location be $S$. It suffices to show that $S$ dominates $L$ (Location
$A$ dominates location $B$ is every path from the entry block to $B$
goes through $A$). We will
prove this using strong induction on the blocks sorted in the topological
order (making it a well founded set). For the sake of convenience, let us
say that $S$ and $L$ are basic blocks. Hence, the statement that we will
prove is the following.
\begin{lemma} \label{lem0}
$P(B) \rightarrow$ $\m{expr}$ can be substituted by 
$\m{v}$ in block $B$ $\Rightarrow$ $S$ dominates $B$.
\end{lemma}

\begin{proof}
Since $S$ dominates itself, $P(S)$ is trivially true. Now, consider $P(B)$.
When we arrive at block $B$ in the second pass, we have already processed
all predecessors of $B$ since we process blocks in the topological order. Let
$t = ComputeHash(\m{expr})$.

\begin{lemma}\label{lem1}
$\m{expr}$ can be substituted by 
$\m{v}$ in block $B$ $\Rightarrow$ $nonNullExprs[B]$ contains $t$.
\end{lemma}

\begin{proof}
Substitution occurs at line~\ref{main-proof1} of Algorithm~\ref{helper},
which can only be reached if line~\ref{lem1-proof} holds.
\end{proof}
Now, for $nonNullExprs[B]$ to contain $t$, $nonNullExprs[blk]$ should also
contain $t$ for all blocks $blk \in B.Preds$, i.e. all predecessors of $B$.
This follows from line~\ref{lem2-proof}
in Algorithm~\ref{gvn-algo}.

\begin{lemma}\label{lem2}
$nonNullExprs[B]$ contains $t \Rightarrow S$ dominates $B$.
\end{lemma}

\begin{proof}
We show Lemma~\ref{lem2} using strong induction on the blocks sorted in
topological order. Clearly, Lemma~\ref{lem2} holds for $S$, as $nonNullExprs
[S]$ contains $t$ and $S$ dominates itself. Since $nonNullExprs[B]$
contains $t$, we know, due to line~\ref{lem2-proof} in Algorithm~\ref{gvn-algo},
that $nonNullExprs[blk]$ contains $t$ for all predecessors $blk$ of $B$. Now, by induction
hypothesis, since the lemma holds for all predecessors of $B$, $S$ dominates
all predecessors of $B$. This implies that $S$ dominates $B$.
\end{proof}
Lemmas~\ref{lem1} and \ref{lem2} together imply Lemma~\ref{lem0},
which is a reformulation of the first condition of the proof
of correctness.
\end{proof}

\subsection{Second Condition}
Let us now prove the second condition.

\begin{lemma}\label{lem3}
If two expressions $\m{e_1}$ and $\m{e_2}$ at locations $L_1$ and $L_2$ respectively
evaluate to the same term $t = ComputeHash(\m{e_1}) = ComputeHash(\m{e_2})$,
then $\m{e_1}$ at $L_1$ and $\m{e_2}$ at $L_2$ evaluate to the same value under
all program executions.
\end{lemma}

\begin{proof}
We prove this lemma by an outer induction on the structure of the expression,
and an inner induction on the blocks sorted in the topological order.
First, we prove this lemma when $\m{e_1}$ and $\m{e_2}$ are both variables. The map
$hashValue$ stores the terms corresponding to each variable in a
particular block. Therefore, $ComputeHash(\m{e_1}) = hashValue[\m{e_1}]$,
which implies $hashValue[L_1][\m{e_1}] = hashValue[L_2][\m{e_2}]$.
Also, by line~\ref{lem3-proof} in Algorithm~\ref{gvn-algo},
we know that for a block $B$, $hashValue[B]$ contains a variable $\m{v}$ only if it evaluates
to the same term in all its predecessors. By the inner induction hypothesis, this
means that $\m{v}$ evaluates to the same value in each predecessor. Also,
whenever a statement of form ``$\m{y := x}$" is encountered, the term for $\m{x}$
is assigned to $\m{y}$. Since this is the only way that two variables can have
the same term, we have the proof of Lemma~\ref{lem3} for variables.

Now, consider the case when $\m{e_1}$ and $\m{e_2}$ are arbitrary expressions.
Suppose $\m{e_1}$ has the form $\m{v_1.f}$, while $\m{e_2}$ has the form
$\m{v_2.g}$. Since $ComputeHash(\m{e_1}) = ComputeHash(\m{e_2})$, we
know that $\m{f} = \m{g}$, and $ComputeHash(\m{v_1}) = ComputeHash(\m{v_2})$.
This is easy to see from the fact that whenever $hashFunction$ is updated
(line~\ref{hf-update} in Algorithm~\ref{helper}), a
new term is added to it. Now, by the outer induction hypothesis, we have that
$\m{v_1}$ and $\m{v_2}$ evaluate to the same value in all executions
and since $\m{f} = \m{g}$,
we have that $\m{e_1}$ and $\m{e_2}$ evaluate to the same value under
all executions. That concludes the proof of Lemma~\ref{lem3}.
\end{proof}

Essentially, Lemma~\ref{lem3} entails that term is an abstract representation
of the value of an expression. Going back to our original proof of correctness,
the variable $\m{v}$ substitutes expression $\m{expr}$ (at 
line~\ref{main-proof1} in
Algorithm~\ref{helper}) only when 
$\m{v} = defaultVar[L][ComputeHash(\m{expr})]$.
Also, $defaultVar$ is updated only when an assignment of the form
$\m{v := expr}$ is encountered (line~\ref{def-asgn}
in Algorithm~\ref{helper}), and before this
update, $ProcessStmt$ is called on the assignment. This
sets $hashValue[S][\m{v}] = ComputeHash(\m{expr})$.
Combining the two arguments above, we have

\[
\begin{array}{rcl}
ComputeHash(\m{v}) \text{ at } S & = & ComputeHash(\m{expr}) \text{ at }
L\\
ComputeHash(\m{v}) \text{ at } L & = & ComputeHash(\m{expr}) \text{ at } L\\
\end{array}
\]

The last line follows from the fact that $\m{v}$ is tagged with $\varsp$, hence
it is assigned only once, and it is available at $L$ due to the first condition in
the proof of correctness. With the two conditions proved, we have that the
transformation introduced in Algorithm~\ref{gvn-algo} produces a semantically
equivalent program, and executing an alias analysis algorithm on the new program
will not add any false positives.

\end{document}